\documentclass[a4paper,11pt]{article}

\usepackage{amsmath, amssymb, amsthm,complexity}
\usepackage{todonotes}
\usepackage[ruled,vlined,linesnumbered]{algorithm2e}
\usepackage{algorithmic}
\usepackage{comment}

\usepackage{thmtools}
\usepackage{thm-restate}
\usepackage{hyperref}

\usepackage{cleveref}
\usepackage{tikz}
\usepackage{amsmath}
\usetikzlibrary{shapes}
\usetikzlibrary{plotmarks}

\usetikzlibrary{arrows,calc,shapes,decorations.pathreplacing}

\usepackage{wrapfig}
\usepackage{lscape}
\usepackage{rotating}
\usepackage{epstopdf}

\usepackage{microtype}

\usepackage{authblk}

\def\margin{2.9cm}
\usepackage[left=\margin,right=\margin,top=\margin,bottom=\margin]{geometry}

\title{A Note on Arc-Disjoint Cycles in Bipartite Tournaments}

\date{}

\author[1]{Jasine Babu}
\author[2]{Ajay Saju Jacob}
\author[1]{R. Krithika}
\author[1]{Deepak Rajendraprasad }

 \affil[1]{Indian Institute of Technology Palakkad, Palakkad, India. \\\texttt{\{jasine|krithika|deepak\}@iitpkd.ac.in}}
 \affil[2]{Indian Institute of Technology Madras, Chennai, India.\\ \texttt{ee16b129@smail.iitm.ac.in}}

\usepackage{complexity} 
\usepackage{enumerate}
\usepackage{amsmath, amssymb, amsthm,complexity}
\usepackage{euler}
\usepackage{authblk}
\usepackage{todonotes}

\def\margin{2.9cm}
\usepackage[left=\margin,right=\margin,top=\margin,bottom=\margin]{geometry}

\usepackage{thmtools}
\usepackage{thm-restate}
\usepackage{hyperref}
\usepackage{cleveref}

\usepackage{authblk}

\usepackage{todonotes}
\usepackage{microtype}
\usepackage{multirow}
\usepackage{amsmath,amssymb}
\usepackage{comment}
\usepackage{enumerate}

\usepackage{xspace}

\usepackage{listings}
\usepackage{color}
\usepackage{cite}
\usepackage{complexity}

\usepackage{graphicx}
\RequirePackage{fancyhdr}
 \usepackage{xcolor}
\usepackage{boxedminipage}

\theoremstyle{plain}
\newtheorem{theorem}{Theorem}

\newtheorem{lemma}[theorem]{Lemma}

\theoremstyle{definition}
\newtheorem{observation}[theorem]{Observation}

\newtheorem{reduction rule}{Reduction Rule}[section]

\DeclareMathOperator{\first}{first}

\makeatletter
\newcommand{\injective}[2][]{\ext@arrow 0359\rightarrowfill@{#1}{#2}}
\makeatother

\begin{document}

\maketitle

\begin{abstract}
We show that for each non-negative integer $k$, every bipartite tournament either contains $k$ arc-disjoint cycles or has a feedback arc set of size at most $7(k-1)$. 
\end{abstract}

\section{Introduction}
Tournaments and bipartite tournaments form a mathematically rich subclass of directed graphs with interesting structural and algorithmic properties \cite{gutin,moon,digraphs}. A tournament is a directed graph obtained by assigning a unique orientation to each edge of an undirected complete simple graph. Similarly, a bipartite tournament is a directed graph obtained by assigning a unique orientation to each edge of an undirected complete bipartite simple graph. Tournaments and bipartite tournaments are tremendously useful in modelling competitions and thus problems on these graphs have several applications in areas like machine learning, voting systems and social choice theory. One such problem is \textsc{Feedback Arc Set}.  A feedback arc set is a set of arcs of the given graph whose deletion results in an acyclic graph. Given a directed graph and a non-negative integer $k$, \textsc{Feedback Arc Set} is the problem of determining if the graph has a feedback arc set of size at most $k$. Finding a minimum feedback arc set in tournaments and bipartite tournaments is \NP-hard \cite{fast-hard-alon,fast-hard,fast-hard3,fasbt-nphard}. However, it is known that for each non-negative integer $k$, every tournament either contains $k$ arc-disjoint cycles or has a feedback arc set of size at most $5k$ \cite{mfcs19} and results from \cite{Chen15,McDonald18} improve the bound of $5k$ to $3.7k$. In this note, we prove an analogous result for bipartite tournaments \footnote{This result is mentioned in \cite{walcom20} as Theorem 2, however, the proof has a gap.}.\\
\noindent {\bf Preliminaries. } A {\em directed graph} (or {\em digraph}) is a pair consisting of a set $V$ of vertices and a set $A$ of arcs. An arc is specified as an ordered pair of vertices. We will consider only simple unweighted digraphs. For a digraph $D$, $V(D)$ and $A(D)$ denote the set of its vertices and the set of its arcs, respectively. Two vertices $u$, $v$ are said to be {\em adjacent}  in $D$ if $(u,v) \in A(D)$ or $(v,u) \in A(D)$. For a vertex $v \in V(D)$, its {\em out-neighborhood}, denoted by $N^{+}(v)$, is the set $\{u \in V(D) \mid (v,u) \in A(D)\}$ and its {\em in-neighborhood}, denoted by $N^{-}(v)$, is the set $\{u \in V(D) \mid (u,v) \in A(D)\}$. This notation is extended to a subset $X$ of vertices as $N^{+}(X)=\cup_{v \in X}N^{+}(v)$ and $N^{-}(X)=\cup_{v \in X}N^{-}(v)$. For a set $X \subseteq V(D) \cup A(D)$, $D-X$ denotes the digraph obtained from $D$ by deleting $X$. 

A {\em path} $P$ in $D$ is a sequence $(v_1,\dots,v_k)$ of distinct vertices such that for each $i \in [k-1]$, $(v_i,v_{i+1}) \in A(D)$. A path $P$ is called an {\em induced} path if there is no arc in $D$ that is between two non-consecutive vertices of $P$. A {\em cycle} $C$ in $D$ is a sequence $(v_1,\dots,v_k)$ of distinct vertices such that $(v_1,\dots,v_k)$ is a path and $(v_k,v_1) \in A(D)$. A cycle $C=(v_1,\dots,v_k)$ is called an {\em induced} (or {\em chordless}) cycle if there is no arc in $D$ that is between two non-consecutive vertices of $C$ with the exception of the arc $(v_k,v_1)$. The length of a path or cycle $X$ is the number of vertices in it and is denoted by $|X|$. A cycle of length $q$ is called a $q$-cycle and a cycle on three vertices is also called a {\em triangle}. A digraph is called a {\em directed acyclic graph} if it has no cycles. Any directed acyclic graph $D$ has an ordering $\sigma$ called {\em topological ordering} of its vertices such that for each $(u,v) \in A(D)$, $\sigma(u)<\sigma(v)$ holds. 

A bipartite digraph is a digraph $B$ whose vertex set can be partitioned into two sets $X$ and $Y$ such that every arc in $B$ has one endpoint in $X$ and the other endpoint in $Y$. We denote $B$ as $B[X,Y]$ where $X$ and $Y$ form the bipartition of the underlying bipartite graph. It is easy to see that a bipartite digraph has no triangle and any 4-cycle is an induced 4-cycle. Given a digraph $D$, $D^R$ denotes the digraph obtained from $D$ by reversing all the arcs. For a set of arcs $F \subseteq A(D)$, $F^R$ denotes the set $\{(u,v) : (v,u) \in F\}$. The following result is well-known.

\begin{observation}
\label{obs1}
A set of arcs $F$ is a feedback arc set of the digraph $D$ if and only if $F^R$ is a feedback arc set of $D^R$.
\end{observation} 

\section{Cycles and Feedback Arc Sets}
In this section, we show that for each non-negative integer $k$, every bipartite tournament either contains $k$ arc-disjoint cycles or has a feedback arc set of size at most $7(k-1)$. Consider a digraph $D$. Let $\mathcal{P}(D)$ denote the set of induced paths on four vertices in $D$. We define two equivalence relations $\sim_2$ and $\sim_3$ on $\mathcal{P}(D)$ as follows. For any two paths $P$ and $P'$ in $\mathcal{P}(D)$, 
\begin{itemize}
\item $P \sim_2 P'$ if and only if paths $P$ and $P'$ differ only in the second vertex.
\item $P \sim_3 P'$ if and only if paths $P$ and $P'$ differ only in the third vertex.
\end{itemize}

For every triple $(x,y,z)$ of vertices in $D$, let $E_{D,2}[x,y,z]$ denote the $\sim_2$-equivalence class consisting of paths in $\mathcal{P}(D)$ with $x$ as first vertex, $y$ as third vertex and $z$ as fourth vertex. Similary, let $E_{D,3}[x,y,z]$ denote the $\sim_3$-equivalence class consisting of paths in $\mathcal{P}(D)$ with $x$ as first vertex, $y$ as second vertex and $z$ as fourth vertex. By definition of $\sim_2$ and $\sim_3$, we have the following observation. 

\begin{observation}
\label{obs2}
For every triple $(x,y,z)$ of vertices in $D$, $E_{D,2}[x,y,z] = E_{D^R,3}[z,y,x]$.
\end{observation}

Next, for a vertex $v \in V(D)$, define the following sets.
\begin{itemize}
\item $\first_D(v)$ is the number of $\sim_2$-equivalence classes consisting of paths in $\mathcal{P}(D)$ with $v$ as the first vertex.
\item $\sec_D(v)$ is the number of $\sim_3$-equivalence classes consisting of paths in $\mathcal{P}(D)$ with $v$ as the second vertex.

\end{itemize}

\begin{observation}
\label{obs3}
$\underset{v \in V(D)}\sum \first_D(v)$ is the number of $\sim_2$-equivalence classes on $\mathcal{P}(D)$ and $\underset{v \in V(D)} \sum \sec_D(v)$ is the number of $\sim_3$-equivalence classes on $\mathcal{P}(D)$.
\end{observation}

Observations \ref{obs2} and \ref{obs3} lead to the following.

\begin{observation}
\label{obs4}
$\underset{v \in V(D)}\sum \first_D(v) = \underset{v \in V(D^R)}\sum \sec_{D^R}(v)$ and $\underset{v \in V(D)}\sum \sec_{D}(v) = \underset{v \in V(D^R)} \sum \first_{D^R}(v)$.
\end{observation}

For a bipartite digraph $D[X,Y]$, let $\Lambda(D)$ denote the number of pairs $u, v$ of vertices in $D$ with $u \in X$, $v \in Y$ and neither $(u,v) \in A(D)$ nor $(v, u) \in A(D)$. Now, we relate the size of a minimum feedback arc set in a bipartite digraph $D$ that has no 4-cycles and the number $\Lambda(D)$. Similar results are known for digraphs that have no 3-cycles \cite{Chen15,Chudnovsky08,Dunkum11} and the proofs crucially use a double counting argument concerning induced paths on three vertices. Our proof for bipartite digraphs with no 4-cycles is along similar lines but involves a more intricate counting argument related to induced paths on four vertices. 

\begin{lemma}
\label{lem:bt-ep}
Let $D[X,Y]$ be a bipartite digraph in which for every pair $u \in X$, $v \in Y$ of distinct vertices, at most one of $(u,v)$ or $(v,u)$ is in $A(D)$. If $D$ has no 4-cycles, then we can compute a feedback arc set of $D$ of size at most $\Lambda(D)$ in polynomial time.
\end{lemma}
\begin{proof}
We will prove the claim by induction on $|V(D)|$. The claim trivially holds for $|X|<2$ or $|Y| < 2$ as in these cases, the empty set is a feedback arc set. Hence, assume that $|X| \geq 2$ and $|Y| \geq 2$. 

First we apply a simple preprocessing rule on $D$. If $D$ has a vertex $v$ that either has no in-neighbours or no out-neighbours then delete $v$ from $D$. As $v$ is not in any cycle of $D$, any feedback arc set of $D'$ is an feedback arc set of $D$. 

\noindent {\bf Case 1:} Suppose $\sum_{v \in V(D)} \first_D(v) \leq \sum_{v \in V(D)} \sec_D(v)$. Then, there is a vertex $u \in V(D)$ such that $\first_D(u) \leq \sec_D(u)$. Without loss of generality assume that $u \in X$. Consider the following sets of vertices of $D$: $Y_1 = N^{-}(u)$, $Y_2 = N^{+}(u)$, $Y_3=Y \setminus (Y_1 \cup Y_2)$, $X_2=N^+(Y_2)$ and $X_1=X \setminus (X_2 \cup \{u\})$. Following are the properties of these sets.
\begin{itemize}
\item $Y_1$ and $Y_2$ are non-empty due to the preprocessing.
\item There is no arc from a vertex $x \in X_2$ to a vertex $y \in Y_1$. Otherwise, $(u,y',x,y)$ is a 4-cycle where $x \in N^+(y')$ and $y' \in Y_2$.
\item By the definition of $X_1$, there is no arc from a vertex $y \in Y_2$ to a vertex $x \in X_1$.
\end{itemize}

Let $D_1$ denote the subgraph $D[X_1, Y_1 \cup Y_3 ]$ and $D_2$ denote the subgraph $D[X_2 \cup \{u\}, Y_2]$. As $D_1$ and $D_2$ are vertex-disjoint subgraphs of $D$, we have $\Lambda(D) \geq \Lambda(D_1) + \Lambda(D_2)$. Further, $\sec_D(u)$ is the number of non-adjacent pairs $a, b$ such that $a \in Y_1$ and $c \in X_2$. As $a \in V(D_1)$ and $c \in V(D_2)$, we have $\Lambda(D) \geq \Lambda(D_1) + \Lambda(D_2) + \sec_D(u)$. 

Let $E$ denote the set of arcs $(x,y)$ such that $x \in X_2$ and $y \in Y_3$.  Let $F_1$ and $F_2$ be feedback arc sets of $D_1$ and $D_2$, respectively. We claim that $F = F_1 \cup F_2 \cup E$ is a feedback arc set of $D$. The sets $F_1$ and $F_2$ are obtained inductively and thus $F$ can be computed in polynomial time. Now, if there exists a cycle $C$ in the graph obtained from $D$ by deleting the arcs in $F$,  then $C$ has an arc $(p,q)$ with $p \in V(D_1)$ and $q \in V(D_2)$ and an arc $(r,s)$ with $r \in V(D_2)$ and $s \in V(D_1)$. Following are the properties of vertices $r$ and $s$.
\begin{itemize}
    \item It is not possible that $r \in Y_2$ and $s \in X_1$ by the definition of $X_1$. 
    \item It is not possible that $r \in X_2$ and $s \in Y_1$ as $D$ has no 4-cycle.
\end{itemize}
Therefore, it follows that $(r,s) \in E$ which leads to a contradiction. Therefore, $F$ is a feedback arc set of $D$ of size $|F|=|F_1|+|F_2|+|E|$. Also, $|E| = \first_D(u)$ and by the choice of $u$, we have $\first_D(u) \leq \sec_D(u)$. Hence, we can conclude that $|F| \leq |F_1|+ |F_2| + \sec(u)$ and by induction hypothesis, $|F_1| \leq \Lambda (D_1)$ and $|F_2| \leq \Lambda (D_2)$. It now follows that $|F| \leq \Lambda (D)$. \\

\noindent {\bf Case 2: } Suppose $\sum_{v \in V(D)} \first_D(v) > \sum_{v \in V(D)} \sec_D(v)$. In this case, from Observation \ref{obs4}, it follows that $\sum_{v \in V(D)} \first_{D^R}(v) < \sum_{v \in V(D)} \sec_{D^R}(v)$. Then, there is a vertex $u \in V(D^R)$ such that $\first_{D^R}(u) \leq \sec_{D^R}(u)$. By a similar argument as in Case 1, it follows that $D^R$ (and hence $D$ from Observation \ref{obs1}) has a feedback arc set of size at most $\Lambda(D^R)=\Lambda(D)$.
\end{proof}

This leads to the following main result.

\begin{theorem}
For every non-negative integer $k$, every bipartite tournament $T$ either contains $k$ arc-disjoint 4-cycles or has a feedback arc set of size at most $7(k-1)$ that can be obtained in polynomial time.
\end{theorem}

\begin{proof}
Suppose $\mathcal{C}$ is a maximal set of arc-disjoint 4-cycles in $T$ with $|\mathcal{C}|\leq k-1$. Let $D$ denote the digraph obtained from $T$ by deleting the arcs that are in some 4-cycle in $\mathcal{C}$. Clearly, $D$ has no 4-cycle and $\Lambda(D) \leq 4(k-1)$. From Lemma \ref{lem:bt-ep}, we know that $D$ has a feedback arc set $F$ of size at most $4(k-1)$. Next, consider a topological ordering $\sigma$ of $D-F$. Each 4-cycle of $\mathcal{C}$ contains at most three arcs which are backward in $\sigma$. If we denote by $F'$ the set of all the arcs of the 4-cycles of C which are backward in $\sigma$, then we have $|F'| \leq 3(k-1)$ and $F \cup F'$ is a feedback arc set of $D$. Therefore, $T$ has a  feedback arc set of size at most $7(k-1)$.
\end{proof}

\end{document}